\documentclass[twocolumn,pra,aps]{revtex4}
\usepackage{amsmath,amssymb,amsthm}
\usepackage{graphicx}
\usepackage{subfigure}
\usepackage{times}\usepackage{mathrsfs}

\usepackage{ulem}
\usepackage{color,hcolor}
\definecolor{earth}{rgb}{0.549,0,0}
\usepackage[bookmarks=false,pdfstartview=FitH,colorlinks,citecolor= earth,linkcolor= earth,urlcolor= earth]{hyperref}

\newcommand{\comment}[1]{}

\newcommand{\ket}[1]{|#1\rangle}

\renewcommand{\Pr}[0]{\operatorname{P}}
\renewcommand{\lg}{\log_2}

\newcommand{\CTRL}[0]{{\operatorname{\textsc{ctrl}}}}
\newcommand{\SIFT}[0]{{\operatorname{\textsc{sift}}}}
\newcommand{\TEST}[0]{{\operatorname{\textsc{test}}}}
\newcommand{\INFO}[0]{{\operatorname{\textsc{info}}}}
\newcommand{\final}[0]{{\operatorname{final}}}
\makeatletter
\def\txtline#1{\noalign{\hbox{\strut\hskip\@totalleftmargin {#1}}}} % no \\ is needed.
\makeatletter
\def\@viiipt{8.5} 
\def\@ixpt{9.5}
\def\@xpt{10.5}
\def\@xipt{11.5}
\def\normalsize{\@setfontsize\normalsize\@xpt{12}}
\def\small{\@setfontsize\small\@ixpt{11}}
\def\footnotesize{\@setfontsize\footnotesize\@viiipt{9.5pt}}%
\def\large{\@setfontsize\large{12.5}{14pt}}%
\makeatother

\newtheorem{prop}{Proposition}
\newtheorem{lemma}[prop]{Lemma}
\newtheorem{theorem}[prop]{Theorem}
\theoremstyle{remark}
\newtheorem{example}{Example}
\theoremstyle{definition}

\begin{document}
 \setlength{\leftmargini}{1.2em}

\title{Semi-Quantum Key Distribution}

\author{Michel Boyer$^{1}$, Ran Gelles$^2$, Dan Kenigsberg$^2$ and Tal Mor$^2$\\
\small 1. D\'epartement IRO, Universit\'e de Montr\'eal,
  Montr\'eal (Qu\'ebec) H3C 3J7 \textsc{Canada} \\
\small 2. Computer Science Department, Technion, 
  Haifa 32000 \textsc{Israel}}

\date{\today}
\begin{abstract}

Secure key distribution among two remote 
parties is impossible when
both are classical, unless some unproven (and arguably unrealistic) 
computation-complexity assumptions are made,
such as the difficulty of factorizing large numbers.
On the other hand, % the above 
a secure key distribution \emph{is} possible when both
parties are quantum.
What \emph{is} possible when only one party (Alice) is quantum, yet the other
(Bob) has only classical capabilities?
Recently, a semi-quantum key distribution protocol was
presented (Boyer, Kenigsberg and Mor, Physical Review Letters, 2007),
in which one of the parties (Bob) is classical, and yet,
the protocol is proven to be completely robust against an
eavesdropping attempt.
Here we extend that result much further.
We present two protocols with this constraint,
and prove their robustness against attacks: we prove
that any attempt of an adversary to obtain information
(and even a tiny amount of  information)  
necessarily induces some errors that the legitimate  parties could notice.
One protocol presented here is identical to the one 
referred to above, 
however, its robustness is proven here in a much
more general scenario.
The other protocol is very different as it is
based on randomization.
\end{abstract}

\maketitle
%TM1 - added "private names" to several parags. 
\section{Introduction}
\label{sec:intro}   
%parag1 - QIP and QKD
Processing information using quantum two-level systems (qubits),
instead of
classical two-state systems (bits), has lead to many striking results
such as the teleportation of unknown quantum states
and quantum algorithms
that are exponentially faster than their known classical counterpart.
Given a quantum computer,  Shor's factoring algorithm would render many of the currently
used encryption protocols completely insecure, 
but as a countermeasure, quantum information processing has also
given
quantum cryptography. 
Quantum key distribution
was invented by Bennett and Brassard (BB84),
to provide a new type of solution to
one of the most important cryptographic problems: the transmission
of secret messages.
A key distributed via quantum cryptography
techniques can be secure even against an eavesdropper with unlimited
computing power,
and the security is % unaffected in the future.
guaranteed forever.

%parag2 - conventional setting of QKD
The conventional setting is as follows:
Alice and Bob have labs that are perfectly secure,
they use qubits for their quantum communication, and
they have access to a classical communication channel
which can be heard, but cannot be jammed
(i.e.\ cannot be tampered with) by the eavesdropper.
The last assumption can easily be justified if Alice
and Bob can broadcast messages, or if they already share some
small number of secret bits in advance, to authenticate the
classical channel.

%parag3 - Here we present!!!
%TM-2g: modified to include classical operations separate from
%       "do nothing"
In the well-known BB84 protocol as well as in 
all other QKD protocols prior to~\cite{BKM07},
both Alice and Bob perform quantum 
operations on their qubits (or on their quantum systems). 
%parag4 - MOTIVATION (ADDED into version 2a)
The question of how much ``quantum'' a protocol  needs to be in order to achieve 
a significant 
advantage over all classical protocols is of great interest.
For example, \cite{Popescu99,JL02,BBKM02,KMR06} discuss whether entanglement is
necessary for quantum computation, \cite{nonloc.noent99} shows nonlocality
without entanglement, and \cite{GPW05,FuchsSasaki03} discuss how
much of the information carried by various 
quantum states is actually classical.
% We extend this discussion into another domain: 
This discussion was extended into the 
{\em quantum cryptography} domain in~\cite{BKM07}
where we presented and analyzed a protocol
in which one party (Bob) is classical. 
For our purposes, 
any two orthogonal states of the quantum two-level system can be chosen  
to be the computational basis $\ket{0}$ and $\ket{1}$.
For reasons that will soon become clear, 
we shall now call the computational basis ``classical'' and we shall use the 
classical notations $\{0,1\}$ to describe  
the two quantum states  
$\{\ket{0},\ket{1}\}$    
defining this basis.
In the protocols we discuss, 
a quantum channel travels from Alice's lab to
the outside world and back to her lab.
Bob can access a segment of the channel, and whenever a qubit 
passes through that segment Bob can either 
let it go undisturbed or % can operate
% on this qubit as follows:
(1) measure the qubit in the classical $\{0,1\}$ basis;
(2) prepare a (fresh) qubit in the classical basis, and send it;
(3) reorder the qubits (by using different delay lines, for instance).
If all parties were limited to performing only
operations (1)--(3), or doing nothing,
they would always be working with qubits in the classical basis,
and could never obtain any
quantum superposition of the computational-basis states; the qubits can then
be considered ``classical bits''; the resulting protocol would then be
equivalent to an old-fashion classical protocol, and
therefore, the operations themselves shall here be considered classical.
% A party that can perform only operations (1)--(3) cannot create quantum superpositions
% from classical bits, i.e., given classical states (density matrices that are diagonal in the
% standard basis) he can only produce classical states.
% The actions of such a limited party is deemed to 
% be classical.
We  term this kind of protocol ``QKD protocol with classical Bob'' %.
or Semi-Quantum Key Distribution (SQKD).
We discuss and analyze two different variants of such a protocol.
In one Bob performs operations (1) and (2) or transfer the qubit back to Alice; 
this variant is therefore named {\em measure-resend SQKD}.
The other variant is based on randomization and named 
{\em randomization-based SQKD}. In this variant Bob is restricted to
perform operations (1) and (3), or do nothing.
%
% Mentioning BKM07 as the pioneering paper.
This work extends the results of~\cite{BKM07}, 
by first generalizing the conditions under which 
the results of~\cite{BKM07} hold for the {\em measure-resend SQKD},
specifically, proving that robustness still 
holds when the qubits are sent one by one and are attacked collectively.
In addition we define and analyze a {\em randomization-based SQKD} 
which leaks no information at all and results with a secret string 
with entropy exponentially close to its length.
We provide a full proof of robustness for this variant as well.

%parag5 - BB84
%TM-2g  modified as Michel suggested (the use of H)
To define our protocols we 
follow the definition (see for instance~\cite{BBBMR06}) 
of the most standard QKD protocol, BB84.
The BB84 protocol consists of two major parts: a first part that is aimed at
creating a \emph{sifted key}, and a second (fully classical) 
part aimed at extracting an error-free,
secure, \emph{final key} from the sifted key. 
In the first part of BB84, Alice randomly selects a binary value and
randomly selects in which basis to send it to Bob, either 
the computational (``$Z$'') basis $\{\ket0,\ket1\}$, or 
the Hadamard (``$X$'') basis $\{\ket+,\ket-\}$.
Bob measures each qubit in either basis at random. 
An equivalent description is obtained if Alice
and  Bob use only the classical operations (1) and (2) above
and the Hadamard\footnote{$H\ket{0}=\ket{+}$;
$H\ket{1}=\ket{-}$.}  quantum gate $H$.
%
%Using the identity $\sigma_X = H\sigma_Z H$ where $H$ is the Hadamard quantum
%gate~\footnote{$H\ket{0}=\ket{+}$; 
%$H\ket{1}=\ket{-}$.},
%an equivalent description is obtained if
% Alice
%and Bob may use only the classical operations (1) and (2) above plus
%the operator $H$.
%
After all qubits have been
sent and measured, Alice and Bob publish which bases they used. For
approximately half of the qubits Alice and Bob used mismatching bases and these
qubits are discarded. The values of the rest of the bits make the sifted key.
The sifted key is identical in Alice's and Bob's hands if  
the protocol is error-free and if there is no eavesdropper (known as Eve)
trying to learn the shared bits or some function of them.
In the second part Alice and Bob use some of the bits of the sifted key 
(the $\TEST$ bits) to test the error-rate, 
and if it is below some pre-agreed threshold, 
they select an $\INFO$ string from the rest of the sifted key. 
Finally, an error
correcting code (ECC) is used to correct the errors on the $\INFO$ string
(the $\INFO$ bits), and privacy amplification (PA)
is used to derive a shorter but unconditionally
secure final key from these $\INFO$ bits.
At that point we would like to mention a key feature relevant
to our protocols:
it is sufficient to use qubits in just one basis, $Z$, 
for generating the $\INFO$ string,
while the other basis is used only for finding the actions of
an adversary~\cite{Mor98}.

%parag6 - security
A conventional measure of security is the information
Eve can obtain on the final key,
and a security proof usually calculates (or puts bounds on)
this information. 
The strongest (most general) attacks allowed by quantum mechanics
are called {\em joint attacks}. 
These attacks are aimed to learn something about the final 
(secret) key directly, 
by using a probe through which all qubits pass, and by measuring
the probe after all classical information becomes public.
Security against all joint attacks is considered as ``unconditional security''.
The security of BB84 (with perfect qubits
sent from Alice to Bob) against all joint attacks was first proven 
in~\cite{Mayers,Shor-Preskill,BBBMR06}   
via various techniques.

\section{Robustness}
%TM1 parag6 - robustness
An important step in studying security is a proof of robustness; see for
instance \cite{BBM92}
for robustness proof of the entanglement-based protocol,
and \cite{SARG04,AGS04} 
for a proof of robustness against the photon-number-splitting attack.
Robustness of a protocol means that any adversarial attempt
to learn some information necessarily induces some disturbance. 
It is a special case (in zero noise) of
the more general ``information versus disturbance" measure which provides
explicit bound on the information available to Eve as a function of the induced
error. 
Robustness also generalizes the no-cloning theorem: while the 
no-cloning theorem states that a state cannot be cloned, robustness means that 
any attempt to make an imprint of a state (even an extremely weak 
imprint) necessarily disturbs the quantum state.

A protocol is said to be
\emph{completely robust} if nonzero information acquired by Eve on the 
$\INFO$ string
implies nonzero
probability that the legitimate participants find errors on the bits tested by the protocol.
A protocol is said to be \emph{completely nonrobust} if Eve can 
acquire the $\INFO$ string without inducing any error on the bits tested by the protocol.
A protocol is said to be
\emph{partly robust} if Eve can acquire some limited information on the 
$\INFO$ string without inducing any error on the bits tested by the protocol.

%TM1 parag7 - robustness scale 
Partially robust protocols could still be secure, yet completely nonrobust 
protocols are automatically proven insecure. 
See also Fig.~\ref{fig:robustness}.
As one example, BB84 is fully robust when qubits are used by Alice and Bob 
but it is only partly robust 
if photon pulses are used and sometimes two-photon pulses
are sent. The well known two-state protocol (also called Bennett92 protocol)
is not fully robust even if perfect qubits are used, if realistic 
channel losses are taken into account.
%TM1 are we interested in these two remarks:
Such partly robust protocols can still lead to a secure final key
if enough bits are sacrified for privacy amplification.
On the other hand, 
such partly robust protocols can become completely nonrobust
(and therefore totally insecure) if the loss rate is sufficiently high.

%%%%%%%%%%%%%%%%%%%%%%%%%%

 \begin{figure}
 \includegraphics[width=\columnwidth]{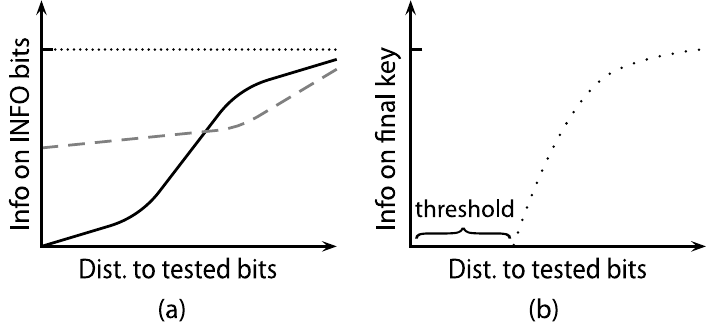}
 \caption{(a) Eve's maximum (over all attacks) information 
  on the $\INFO$ string vs. the allowed disturbance on the bits tested by
  Alice and Bob, in a completely robust (solid line), partly robust
  (dashed), and completely nonrobust (densely dotted) protocol.
  (b) Robustness should not be confused with security; Eve's maximum 
  information on the \emph{final} key vs. allowed disturbance in a secure
  protocol; such a protocol could be completely or partly robust.}
  \label{fig:robustness}
  \end{figure}
 
%%%%%%%%%%%%%%%%%%%%%%%%%%

\section{Mock protocol and its complete nonrobustness}
\label{sec:idea}

%TM1 parag9 - Mock  
Consider the following mock protocol:
Alice flips a coin to decide whether to send a random bit in the computational
basis $\{\ket{0},\ket{1}\}$ (``$Z$''), or in the Hadamard basis
$\{\ket{+},\ket{-}\}$ (``$X$''). 
Bob flips a coin to decide whether to measure Alice's qubit in the computational
basis (to ``$\SIFT$'' it) or to reflect it back (``$\CTRL$''), without causing any
modification to the information carrier.
In case Alice chose $Z$ and Bob decided to $\SIFT$, i.e.\ to measure in the $Z$
basis, they share a random
bit that we call $\SIFT$ or sifted bit (that may, or may not, be confidential). 
In case Bob chose $\CTRL$, Alice can check if the qubit returned
unchanged, by measuring it in the basis she sent it. 
In case Bob chose to $\SIFT$ and Alice chose the $X$ basis,
they discard that bit.
The idea that just one basis, the $Z$-basis, is sufficient for the key
generation (while the other basis is used for finding the actions of
an adversary) appeared already in \cite{Mor98}.
The above iteration is repeated for a predefined number of times. 
At the end of the quantum part of the protocol Alice and Bob share, 
with high probability, a considerable amount of $\SIFT$ bits (also known
as the ``sifted key'').
In order to make sure that Eve cannot gain much information by measuring
(and resending) all
qubits in the $Z$ basis, 
Alice can check whether they have a low-enough level of discrepancy
on the $X$-basis  $\CTRL$ bits. 
In order to make sure that their sifted key is reliable, 
Alice and Bob must sacrifice
a random subset of the $\SIFT$ bits, which we denote as $\TEST$ bits, 
and remain with a string of bits which we call $\INFO$ bits 
($\INFO$ and $\TEST$ are common in QKD, e.g., in BB84 as previously described).

%TM1 parag10 - Mock ECC+PA  
By comparing the value of the $\TEST$ bits,
Alice and Bob can estimate the error rate on the $\INFO$ bits.
If the error rate on the $\INFO$ bits 
is sufficiently small, 
% Alice and Bob 
they can then use an appropriate Error Correction Code (ECC) 
in order to correct the errors.
If the error rate on the $X$-basis $\CTRL$ bits  
is sufficiently small, Alice and Bob can bound Eve's information,
and can then use an appropriate Privacy
Amplification (PA) in order to 
obtain any desired level of privacy.

%TM1 parag11 - Mock PROBLEM 
At first glance, 
this protocol may look like a nice way to transfer 
a secret bit from quantum Alice to classical Bob: It is
probably resistant to opaque (intercept-resend) attacks.

However, it is {\em completely non-robust}; 
%In a rather simple fashion, 
Eve could learn all bits of the $\INFO$ string
using a trivial %joint 
attack that induces no error on the bits tested by Alice 
and Bob (the $\TEST$ and
$\CTRL$ bits).
She would not measure the incoming qubit, but rather perform a cNOT from 
it into a $\ket{0^E}$ ancilla\footnote{By the term ``cNot from
$A$ into $B$'' we mean that $A$ is the control qubit and $B$ is the
target, as is commonly called; we prefer to use the term ``control qubit'' 
in a different meaning in our paper.}.
If Alice chose $Z$ and Bob decides to $\SIFT$ (i.e.\ measures in the $Z$-basis), 
she measures her ancilla and
obtains an exact copy of their common bit, thus inducing no
error on $\TEST$ bits and learning the $\INFO$ string.
If, however, Bob decides on $\CTRL$, i.e.\ reflects the qubit,
Eve would do another cNOT from the returning qubit into her ancilla. 
This would reset her ancilla, erase the interaction she performed, and
induce no error on $\CTRL$ bits, thus
removing any chance of her being caught. 
In the following Section we present two protocols which overcome this problem
via two different methods.

\section{Two Semi-Quantum Key Distribution Protocols}
% \label{sec:protocol}
%TM1 parag12 - Two protocols for SQKD
The following two protocols remedy the above weakness
by not letting Eve know which is
a $\SIFT$ qubit (that can be safely measured in the computational basis)
and which is a $\CTRL$  qubit (that should be returned to Alice unchanged).
Both protocols are aimed at creating an $n$-bit $\INFO$ string to be
used as the seed for an $l$-bit shared secret key.

%\medskip
\subsection{Protocol~1: Randomization-based SQKD.}

Two versions are presented, both based on randomizing the returned qubits: Protocol~1 depends on a single
parameter $\delta > 0$ and is not completely-robust; Protocol~1$'$, 
with an additional parameter $\epsilon\leq 1$ 
such that $0\leq \epsilon < \delta$
and with Step~7$'$ replacing Step~7, 
is completely robust. 

Let $n$, the desired length of the $\INFO$ string, be an even integer and
let $\delta>0$ be
some fixed parameter.
\begin{enumerate}
\item Alice sends $N = \lceil8n(1+\delta)\rceil$ qubits.
For each of the qubits
she randomly selects whether to send it
in the computational basis ($Z$) or the Hadamard
basis ($X$). 
In each basis she sends random bits.
\item For each qubit arriving, Bob chooses randomly whether to measure it (to $\SIFT$ it)
or to reflect it ($\CTRL$). Bob reorders randomly
the reflected qubits so that no one, neither Alice nor Eve, 
could tell which of them were
reflected.
\label{step:bob-choice}
\item Alice collects the reflected qubits in a quantum memory\footnote{Quantum
memory is not strictly required, since instead of it (with a certain  % 1/2
penalty to the protocol rate) Alice can 
measure each reflected qubit in a random basis.}.
\label{step:alice-collect}
\item Alice publishes which were her $Z$ bits. 
Bob publishes which were his $\CTRL$ qubits, and in which order they
were reflected;
Alice then measures all the returned $\CTRL$ qubits in the  
basis she prepared them.
\label{step:alice-compare}
\end{enumerate}
It is expected that for approximately $N/4$ bits, Alice used the $Z$ basis
and Bob chose to $\SIFT$ (these are the $\SIFT$ bits, which form the sifted key);
for approximately $N/4$ bits,
Alice used the $Z$ basis
and Bob chose $\CTRL$ (we refer to these bits as $Z$-$\CTRL$),
and for approximately $N/4$ bits,
Alice used the $X$ basis
and Bob chose $\CTRL$ (we refer to these bits as $X$-$\CTRL$).
In the rest of the bits, Bob expects a uniform distribution.
Cf. Fig~\ref{fig:protsum}.
\begin{enumerate}
\setcounter{enumi}{4}
\item
Alice checks the error-rate on the $\CTRL$ bits and if either
the $X$ error-rate or the $Z$ error-rate is higher than some
predefined threshold $P_{\CTRL}$ the protocol aborts.
\item Alice chooses
at random $n$ $\SIFT$ bits to be $\TEST$ bits.
She publishes which are the chosen bits.
Bob publishes the value of these $\TEST$ bits. 
Alice checks the error-rate on the $\TEST$ bits and if it is higher than some
predefined threshold $P_{\TEST}$ the protocol aborts. Else, let $v$ be the
string of the remaining $\SIFT$ bits.
\end{enumerate}
\begin{enumerate}
\setcounter{enumi}{6}
\item
Alice and Bob select the first $n$ 
bits in $v$ to be used as $\INFO$ bits.
%RG13b - there is ambiguity when saying INFO, since Bob and Alice possible share
%        different strings. (e.g., when we say Eve has no information on the INFO
%        we mean "The info String of Alice", etc.
If there is no errors or eavesdropping, 
Alice and Bob share the same string.
Otherwise, Bob's string is likely to differ from the the $\INFO$ 
string until corrected in Step~8 below.
\end{enumerate}
Unfortunately, Protocol 1 is not robust:
we will show how Eve can
count the number of ``$0$''s and ``$1$''s measured by Bob
(i.e.\ the Hamming weight of the measured string)
 without being
detectable and
get about $0.3$ bits of information on the $\INFO$ string, whatever
its length (and prove she can not do better). 
\begin{figure}
\begin{center}
\begin{tabular}{c@{~~}|@{~~}c@{~~}|@{~~}l@{~~}|@{~~}l}
A&B& Name &  Usage\\\hline
$Z$ &$\SIFT$&$\SIFT$& A pool for $n$ $\INFO$ and $n$ $\TEST$ bits\\
$X$ &$\SIFT$& &Bob expects a uniform distribution \\
$Z$ &$\CTRL$& $Z$-$\CTRL$ & Alice expects her values unchanged \\
$X$ &$\CTRL$& $X$-$\CTRL$ & Alice expects her values unchanged
\end{tabular}
\end{center}
\caption{Bit usage summary}
\label{fig:protsum}
\end{figure}

To make sure Eve cannot use statistics of occurrence of  
``$0$''s and ``$1$''s in the $\INFO$
string, Protocol~$1'$ will fix in advance a subset of $\{0,1\}^n$ to be used
for the $n$-bit $\INFO$ strings. 
A new parameter $\epsilon \leq 1$ such that $0\leq \epsilon < \delta$
is introduced and the set of  $\INFO$ strings is 
\begin{equation}\label{Inepsilondef}
I_{n,\epsilon} = \left\{ y \in \{0,1\}^n \mid  \left| \frac{|y|}{n} 
   -\frac{1}{2}\right| \leq\frac{\epsilon}{2} \right\} 
\end{equation}
where $|y|$ denotes the Hamming weight of $y$. 
When $\epsilon = 0$, $I_{n,0}$ is the set of $n$-bit 
strings with Hamming weight $n/2$;
for $\epsilon = 1$ (which can happen if $\delta > 1$), $I_{n,\epsilon} = \{0,1\}^n$.
We will prove that when $\epsilon > 0$, the information carried by a random $y \in I_{n,\epsilon}$
is exponentially close to $n$ bits (in the parameter $n$).
In that case, 
the set $I_{n,\epsilon}$ %when $\epsilon > 0$ is thus 
is a ``good set'' of $\INFO$ strings. 
When $\epsilon = 0$, $I_{n,0}$ has entropy of the order $n-0.5\log_2(n)$ bits. 

As for robustness, it is obtained by replacing Step 7 by Step 7$'$:
\begin{enumerate}
\setcounter{enumi}{6}
\item[7$'$\negthinspace.] %If the protocol does not abort  
\begin{enumerate}
\item Alice chooses a substring $x$ of $v$ of length $2h$ with
$h$ zeros and $h$ ones, where $h = \lfloor (1+\epsilon)n/2\rfloor$;
if she can not choose such a string, the protocol aborts.
\item Alice chooses randomly $y\in I_{n,\epsilon}$.
\item Alice chooses randomly a list of distinct indices $q_1\ldots q_n$ such that $x_{q_1}\ldots x_{q_n} = y$.
\item Alice announces publicly $q_1\ldots q_n$; Bob thus learns that $v_{q_1}\ldots v_{q_n}$ is the $\INFO$ string.
\end{enumerate}
\end{enumerate}
We will show that the protocol aborts with exponentially small probability and leaks no
information to Eve as long as she is undetectable.

\subsection{Protocol~2: Measure-Resend SQKD}

%Classical Bob --- the obliviousness-based protocol:

Our second protocol does not require Bob
to randomize the qubits as in Step~\ref{step:bob-choice}.
Instead, Bob either measures and resends the qubit ($\SIFT$s it)
or reflects it ($\CTRL$).
Furthermore, Alice does not need to delay the measurement of
the returning qubits until Step~\ref{step:alice-compare},
because immediately in Step~\ref{step:alice-collect} she knows in which
basis to measure.
 
The protocol is essentially the same
as the previous one, with steps 1 to 7,  %RG12b - step 8 is presented only below.. 
but with steps 2, 3 and 4 modified
to correspond to the new simplified sifting procedure; the modified steps are:
\begin{enumerate}
\item[2.] For each qubit arriving, 
Bob chooses randomly whether to measure 
and resend it in the same state he found (to $\SIFT$ it)
or to reflect it ($\CTRL$). 
Again, no one, neither Alice nor Eve, 
can tell which of the qubits were
reflected.
\item[3.] Alice measures each 
qubit in the basis she sent it.
\item[4.] Alice publishes which were her $Z$ bits and Bob publishes
which ones he chose to $\SIFT$.
\end{enumerate}

\subsection{Classical Post-Processing.}

The full protocol for the generation of the final key comprises 
any one of the above
``semi-quantum'' protocols, plus the ``classical'' step:
\begin{enumerate}
\setcounter{enumi}{7}
\item Alice publishes ECC \& PA data, from
which she and Bob extract the $l$-bit final key from the % $n$-bit
$\INFO$ string. % \MB{($l$ being the number of PA strings)}.
\end{enumerate}
If the ECC is of rank $R$, 
publishing the ECC data entails publishing the parities of $n-R$ substrings
of the $\INFO$ string, i.e.\ up to $n-R$ bits of information on the $\INFO$ string.
This step must thus be excluded from the definition of
robustness or else no protocol would ever be robust unless the ECC is 
degenerate (of rank $n$) and unable to correct any error 
(the minimal distance being 1).
The $l$-bit key is chosen such that $l \leq R$ and it is the information
on this final $l$-bit key that needs to 
be proven negligible to prove the security of the above protocols.

\section{Proofs of Robustness}

%Can Eve mount an attack on the obliviousness-based protocol, and
%learn information about INFO bits without being caught?
We first show that Eve cannot obtain information on $\INFO$ bits in Protocol
2 without being detectable 
for the case in which the qubits are sent by Alice one by one as well as
as the case they are sent together. This is performed by considering the general case
in which Alice sends the qubits one by one but does not wait for a returning qubit before sending the
next one (so that Eve can collect the qubits and attack them collectively). The scenario 
analyzed in~\cite{BKM07} is a specific case of the setup we analyze here.
%\MB{under the assumption that qubits are sent one by one,
%are expected by Bob in the order they are sent but Alice does not wait for a returning
%qubit before sending the next one; Eve could thus attack collectively all qubits on their
%way back before returning them to Alice, a setup that is more general than the one
%considered in~\cite{BKM07}}. 
We then 
bound the information Eve can get with Protocol~1 without inducing errors on
$\TEST$ and $\CTRL$ bits and finally 
prove the complete robustness of Protocol~1$'$.

\subsection{Complete Robustness of Protocol~2.}
\label{sec:robust}

\subsubsection{Modeling the protocol.} Each time the protocol is executed,
Alice sends to Bob a state $\ket{\phi}$ which is a tensor product of $N$ qubits, each of which
is either $\ket{+}$, $\ket{-}$, $\ket{0}$ or $\ket{1}$; 
those qubits are indexed  from $1$ to $N$. 
Each of those qubits is either measured
by Bob in the standard basis and resent as it was measured or simply reflected. 
We denote $m$ the set of bit positions 
measured by Bob; this is a subset of $[1\,\ldots\,N]$ that we represent by
an increasing list of $r$ integer positions $m_1\ldots m_r$ corresponding to Bob
measuring the
$r$ qubits with index $m_1$, \ldots, $m_r$. For $i\in \{0,1\}^N$, we denote
\[
i_m = i_{m_1}i_{m_2}\ldots i_{m_r}
\]
 the substring of $i$ of length $r$ selected by
the positions in $m$;
of course $\ket{i_m} = \ket{ i_{m_1}i_{m_2}\ldots i_{m_r}}$.

In the protocol, it is assumed that Bob has no quantum register; 
he measures the  qubits as they come in. The physics
would however be exactly the same if Bob used a quantum register of $r$ qubits initialized
in state $\ket{0^B} = \ket{0^r}$ ($r$ qubits equal to $0$), applied the unitary transform
defined by\footnote{If $\ket{j}_B$ is Bob's register with $j\in \{0,1\}^r$,
then $M_m\ket{i}\ket{j}_B = \ket{i}\ket{i_m\oplus j}_B$
where $\oplus$ denotes a bitwise exclusive or.}
\begin{equation}
M_m \ket{i}\ket{0^B} = \ket{i}\ket{i_m}
\end{equation}
for $i\in\{0,1\}^N$, sent back $\ket{i}$ to Alice
 and postponed his measurement to be performed on that quantum register $\ket{i_m}$;
 the qubits indexed by $m$ in $\ket{i}$ are thus automatically both measured and
 resent, and those not in $m$
 simply reflected;
 the $k$th qubit sent by Alice is a $\SIFT$ bit if $k\in m$ and is either $\ket{0}$ or $\ket{1}$;
it is a $\CTRL$ bit if $k\notin m$.
This physically equivalent modified protocol
 simplifies the analysis and we shall thus model Bob's measurement and resending, or
 reflection, with $M_m$.  In most cases, Bob's measurement will be performed bitwise;
 for each $k$ in $m$ we will denote $M_k$ the unitary that performs an exclusive or
 between $k$-th qubit
 in $i$ and on the corresponding qubit $j_k$ in Bob's probe i.e.\ $M_k\ket{i_k}\ket{j_k} =
 \ket{i_k}\ket{j_k\oplus i_k}$. It follows that
 \[
 M_m = M_{m_r}\ldots M_{m_2}M_{m_1}.
 \]

\subsubsection{Eve's attack.}
The special case where all qubits go from Alice to Bob before coming back, which happens if they are
sent in parallel, was analyzed in \cite{BKM07}.
Eve's  most general attack is then comprised of two unitaries: $U_E$ attacking qubits as they
go from Alice to Bob and  $U_F$ as they go back from Bob to Alice, where $U_E$ and
$U_F$ share a common probe space with initial state $\ket{0^E}$.
The shared probe allows Eve to make the attack on the returning qubits depend
on knowledge acquired by $U_E$ (if Eve does not take advantage of that fact,
then the ``shared probe'' can simply be the composite system comprised of two 
independent probes). Any attack where Eve would make $U_F$ depend
on a measurement made after applying $U_E$ can be implemented by a unitaries
$U_E$ and $U_F$ 
with controlled gates
so as to postpone measurements; since we are giving Eve all the power of quantum
mechanics, the difficulty of building such a circuit is of no concern.
Eve can use at will a general-purpose quantum computer. 

The following (more general attack) is
possible if  Bob is expecting qubits in a sequence
yet Alice does not wait for a returning qubit before sending the next one.
Since Eve has access to a quantum memory, she can wait till she gets
all qubits $\ket{\phi}$ sent by Alice before proceeding.
Once she got them all, the most general attack she
can perform applies a unitary transform to $\ket{0^E}\ket{\phi}$, 
sends the first qubit to Bob, waits till it
comes back from Bob to then repeat the 
same action (with a possibly different unitary each time) 
for each qubit in a sequence. 
When Eve has attacked all qubits forth and back, 
she sends them back to Alice (one by one if needed).

More formally let $\mathscr{H}_P = \bigotimes_{k=1}^N \mathscr{H}_k$ be the space of the protocol, 
where each $\mathscr{H}_k$ is the two dimensional Hilbert space corresponding to the $k$-th qubit and 
let $\mathscr{H}_E$ be Eve's probe space; once Eve holds $\ket{\phi}$ she 
applies a unitary $U_1$ on $\ket{0^E}\ket{\phi}$ and sends Bob qubit 1
(corresponding to $\mathscr{H}_1$).
For each qubit $k$ from $1$ to $N-1$,
when Eve receives qubit $k$ back from Bob, she applies $U_{k+1}$ on $\mathscr{H}_E\otimes \mathscr{H}_P$
and then sends qubit $k+1$ to Bob. When Eve receives qubit $N$ from Bob, she applies $U_{N+1}$
on $\mathscr{H}_E\otimes\mathscr{H}_P$, sends the $N$ qubits to Alice and keeps her probe.
Eve's attack is thus characterized by a sequence $\{U_k\}_{1\leq k \leq N+1}$ of unitary transforms on
$\mathscr{H}_E\otimes\mathscr{H}_P$.

The attack in \cite{BKM07} where  Eve applies $U_E$ to all qubits, sends them to Bob, and applies $U_F$ on their
way back corresponds to the attack where  $U_1 = U_E$, $U_2 = \ldots = U_{N} = I$ and 
$U_{N+1} = U_F$ i.e.\ Eve uses $U_E$ on all qubits when she receives them, does nothing till she got
all qubits back and then applies $U_F$.

Another prococol, whose robustness can be proved 
with the methods of \cite{BKM07} and which is briefly % rapidly  
mentioned in its conclusion
requires each qubit to be sent individually, Alice sending each qubit only when she
received the previous one from Bob.
Eve also uses a global probe initialized to $\ket{0^E}$ but she is forced to attack
qubits individually. For each qubit $k$ 
from 1 to $N$, Eve applies a unitary $U_E^{(k)}$ acting on  
$\mathscr{H}_E$ and $\mathscr{H}_k$\footnote{The transforms $U_E^{(k)}$ and $U_F^{(k)}$ act on
 $\mathscr{H}_E\otimes \ldots \otimes \mathscr{H}_k\otimes \ldots$ and leave
$\mathscr{H}_l$ fixed for $l\neq k$. } before sending it to Bob 
and applies a unitary $U_F^{(k)}$ acting on the same spaces
on the way back. The robustness of the individual-qubit protocol follows immediately from 
the robustness of Protocol 2 under the limited class of attacks where
$U_1 = U_E^{(1)}$,
$U_k = U_E^{(k)}U_F^{(k-1)}$ for $1\leq k < N$ and $U_{N+1} = U_F^{(N)}$ 
(and qubits are returned all together to Alice).

\subsubsection{The final global state.} 
Delaying all measurements allows considering the global state
of the Eve+Alice+Bob system before all actual measurements;
Eve's and Bob's actions are described by unitary transforms.
The initial state is $\ket{0^E}\ket{\phi}\ket{0^B}$; Eve's unitary
transforms $U_1$, \ldots, $U_{N+1}$ act on the first two Hilbert spaces whilst 
Bob's measurements $M_k$ performed when he receives qubit $k$ with $k\in m$
act on the last two spaces. For instance, if $N=4$ and $m=(1,3)$ then the final global
state of the system is $U_5U_4M_3U_3U_2M_1U_1\ket{0^E}\ket{i}\ket{0^B}$
where measurement $M_1$ on qubit 1 occurs immediately after Eve applies $U_1$ and measurement
$M_3$ on qubit 3 occurs immediately after Eve applies $U_3$.
 
The attacks $\{U_k\}_{1\leq k \leq N+1}$ we are interested in are only those
for which Eve is completely undetectable. Such attacks put strong restrictions
on the global evolution of the system. In what follows, when we say that an attack
 induces no error on $\CTRL$ and $\TEST$, we mean
that for any choice of $\CTRL$ and $\TEST$  bits 
whose probability of occurrence according to protocol 2 is not 0, the probabililty
that Eve's attack induces an error on them is $0$.

\begin{prop}\label{propositiondefeideffi}
If the attack $\{U_k\}_{1\leq k \leq N+1}$ induces no error on $\TEST$  and $\CTRL$ bits,
and if Alice sent state $\ket{i}$ with $i\in \{0,1\}^N$, then there is a state 
$\ket{F_i} \in \mathscr{H}_E$ such that, for all $m$, the final global state of the system
after applying $U_{N+1}$ is
\begin{equation}\label{summerize0}
\ket{F_i}\ket{i}\ket{i_m}.
\end{equation}
\end{prop}
\begin{proof}

The final global state of the system can always be written as
$
\sum_{jj'} \ket{E_{ijj'}} \ket{j}\ket{j'}
$
where $\ket{j}$ is the standard basis of $\mathscr{H}_P$ and $\ket{j'}$ of Bob's probe space;
If the protocol induces no errors on $\TEST$ bits, it must be so that for all $m$, 
$\ket{E_{ijj'}} = 0$ for $j' \neq i_m$ and thus the final global state must be
$
\sum_j \ket{E_{iji_m}}\ket{j}\ket{i_m}.
$
Moreover, if there is no error on $\CTRL$ bits, then the probability for Alice to
measure any $\ket{j}$ that is not $\ket{i}$ must be zero. She can indeed choose
any qubit not in $m$ as a {$Z$-$\CTRL$} bit; she also checks all the qubits measured
by Bob, which must also coincide with those she sent since $i\in \{0,1\}^N$.
Consequently $\ket{E_{iji_m}} = 0$ if $j\neq i$ and the final state must be
$
\ket{E_{iii_m}}\ket{i}\ket{i_m}.
$

We now prove that $\ket{E_{iii_m}}$ does not depend on $i_m$. Let $Z$ be the linear
map defined by $Z\ket{e}\ket{j}\ket{j'} = \ket{e}\ket{j}\ket{0^B}$ i.e.\ $Z$ is the linear  map on 
Bob's probe space that maps its standard basis states on the state $\ket{0^B}$.
It is clear that $ZU_k = U_kZ$ and $ZM_k = Z$ for all $k$.
If we look at the particular case where $N=4$ and 
$m=(1,3)$, i.e.\ Bob measures qubits 1 and 3, this implies that
$ZU_5U_4M_3U_3U_2M_1U_1\ket{0^E}\ket{i}\ket{0^B} =
U_5U_4U_3U_2U_1Z\ket{0^E}\ket{i}\ket{0^B} = U_5U_4U_3U_2U_1\ket{0^E}\ket{i}\ket{0^B}$.
Applying $Z$ to the final state just gives the final state obtained if $m$ is empty.
If we apply $Z$ to $\ket{E_{iii_m}}\ket{i}\ket{i_m}$ we get $\ket{E_{iii_m}}\ket{i}\ket{0^B}$ and
this state must be equal to the final global state when $m$ is empty. This implies that for all
values of $m$, the states $\ket{E_{iii_m}}$ must be the same; we call them $\ket{F_i}$
and this gives $
\ket{F_i}\ket{i}\ket{i_m}
$
as the final global state.
Note that the Eve's state $\ket{F_i}$ is not entangled with the system $\ket{i}$ sent back to Alice,
nor with Bob's register $\ket{i_m}$.
\end{proof}

We now show that if Eve's attack is undetectable by Alice and Bob, then
Eve's final state $\ket{F_i}$ 
is independent
of the string $i\in \{0,1\}^N$. More precisely,
\begin{prop}\label{propprot2}
If  $\{U_k\}_{1\leq k \leq N+1}$ is an attack on Protocol~2  that
 induces no error on $\TEST$ and $\CTRL$ bits, then
  for all $i, i' \in \{0,1\}^N$,
 \begin{equation}
 i,i'\in \{0,1\}^N \quad \implies \quad \ket{F_i} = \ket{F_{i'}}.
 \end{equation}
\end{prop}
\begin{proof}
For any index $k$, 
%RG12b - writting it clearly
let Alice's $k$-th qubit be in state $\ket{+}$, and all the other qubits
be prepared in the $Z$-basis.
Alice's state can be written $\frac{1}{\sqrt{2}}[\ket{i} + \ket{i'}]$
where  $i, i' \in \{0,1\}^N$, 
$i_k=0$, $i'_k=1$, and $i_t = i'_t$ for $t\neq k$. 
Let Bob choose $m$ such that $k\notin m$; such an $m$ exists because $N\geq 2$
and then $i_m = i'_m$. 
By the previous proposition and linearity, the final global state is
$\frac{1}{\sqrt{2}}\left[ \ket{F_i}\ket{i} + \ket{F_{i'}}\ket{i'}\right]\ket{i_m}$;
since we are interested only in % by 
Alice's $k$-th qubit, we trace-out all the other qubits
in Alice and Bob's hands and get the state
\[
\frac{1}{\sqrt{2}}\left[\ket{F_i}\ket{0} + \ket{F_{i'}}\ket{1} \right];
\]
if $\ket{0}$ and $\ket{1}$ are 
replaced by their values in term of $\ket{+}$ and $\ket{-}$, this
rewrites
$
\frac{1}{2}\Big[\ket{F_i}+\ket{F_{i'}}\Big] \ket{+} + \frac{1}{2}\Big[\ket{F_i}-\ket{F_{i'}}\Big]\ket{-}
$
and since the probability that Alice measures $\ket{-}$ must be~$0$, 
$\frac{1}{2}\Big[\ket{F_i}-\ket{F_{i'}}\Big]=0$ i.e.\
$\ket{F_i} = \ket{F_{i'}}$.
The above holds for any $l$; 
any bit in $i$ can be flipped without affecting $\ket{F_i}$ and thus
$\ket{F_i}$ is the same for all $i\in \{0,1\}^N$.
\end{proof}

\begin{theorem}\label{thmrobustness2}
 For any attack $\{U_k\}_{1\leq k\leq N+1}$ on Protocol~2 that induces no error
on $\TEST$ and $\CTRL$ bits, Eve's final state is independent of the state
$\ket{\phi}$ sent by Alice, and Eve has thus no information on the $\INFO$ string.
\end{theorem}
\begin{proof}
By Proposition~\ref{propprot2}, there is a state $F_\final$ of Eve's probe space
such that for all $i\in \{0,1\}^N$, Eve's final state
$\ket{F_i} = \ket{F_\final}$. By Proposition~\ref{propositiondefeideffi},
for all $i\in \{0,1\}^N$ and all $m$, the final state after applying $U_{N+1}$ if Alice sends 
$\ket{i}$ is thus
$\ket{F_\final}\ket{i}\ket{i_m}$. For all superpositions $\ket{\phi} = \sum_i c_i\ket{i}$ that
Alice may send, and all $m$, the final state of the
Eve+Alice+Bob system after applying $U_{N+1}$ is consequently
\begin{equation}
\ket{F_\final} \sum_i c_i\ket{i}\ket{i_m};
\end{equation}
Eve's probe state  $\ket{F_\final}$ is independent of % $\ket\phi$ and $m$, 
$i_m$ and therefore of the $\SIFT$ bits and $\INFO$ bits --- if Eve is to be undetectable.
\end{proof}

%% RG7a - protocol 1 --> protocol 2.
The above theorem means that Protocol~2 is completely robust.

%RG10a- sections elabortaion
\subsection{Partial robustness of Protocol~1.}

\subsubsection{Modeling the protocol.} The states $\ket{\phi}$ sent by Alice are still
products of $N$ qubits each of which is either $\ket{+}$, $\ket{-}$, $\ket{0}$ or $\ket{1}$.
In Step~2 of the protocol, Bob either measures a qubit, or reflects it; moreover, he reorders
randomly the reflected qubits; let $r$ be the number of reflected qubits and let
$s = s_1s_2\ldots s_r$ be the list of those $r$ randomly ordered bit positions.
For instance, if $r=4$, and Bob reflects qubits $8$, $1$, $5$ and $4$ in that order then
$s = 8154$ (examples will use positions from $1$ to $9$ to avoid comma separated
lists). The list of non-reflected bits is indexed by the complement $\bar{s}$ and will always
be listed in ascending order; if $N=9$ and $s=8154$ then $\bar{s} = 23679$.
Bob's measurement can still be postponed, but this time, since Bob keeps the qubits selected
by $\bar{s}$ without sending a copy, there is no need to copy. 
For all string $s$ we still denote $i_s = i_{s_1}\ldots i_{s_r}$ the list of bits selected by $s$
in the order specified by $s$; 
Bob's operation is then captured by
\[
U'_s\ket{i} = \ket{i_s}\ket{i_{\bar{s}}}
\]
where $\ket{i_s}$ is the state reflected to Alice, and $\ket{i_{\bar{s}}}$ the state (to be) measured
by Bob. With $N=9$ and $s=8154$, and if Alice sent
$\ket{i_1\ldots i_9}$ with $i_1,\ldots, i_9\in \{0,1\}$,
the state reflected is $\ket{i_8i_1i_5i_4}$ and the
state to be measured $\ket{i_2i_3i_6i_7i_9}$. Of course, Alice can compare $i_s$ with
what she actually sent only when $s$ is known and consequently keeps $\ket{i_s}$
in quantum memory. With these notations, qubit $k$ is $\CTRL$ if $k\in s$ and it is $\SIFT$ if
it is either $\ket{0}$ or $\ket{1}$ and $k\notin s$.

\subsubsection{Eve's attack.}
Eve's  most general attack is still comprised of two unitaries: $U_E$ and $U_F$
sharing a common probe space; $U_E$ is applied on
$\ket{0^E}$ and $\ket{\phi}$ and attacks qubits as they
go from Alice to Bob;  $U_F$ is applied on Eve's probe and $\ket{i_s}$ as 
those bits go back from Bob to Alice; one slightly annoying problem is that the dimension
of the space on which $U_F$ acts is not fixed; it depends on the size of $s$, i.e.
the number of bits reflected by Bob; there is thus one unitary 
$U_F$  for each $r>0$.

\subsubsection{The global final state.}
Since Bob uses no probe space, the global state after Eve applies $U_E$ is simply
$U_E\ket{0^E}\ket{\phi}$; then Bob applies $U'_s$ to his part of the system, which corresponds
to the global unitary $I_E\otimes U'_s$ where $I_E$ is the identity on Eve's probe space.
Then $U_F$ is applied only on Eve's probe and $\ket{i_s}$; if we denote $I_{\bar{s}}$ the
identity on the system left in Bob's hands, given by the qubits selected by $\bar{s}$, the
final global state is then
\begin{equation}\label{globalproto1}
[U_F\otimes I_{\bar{s}}][ I_E\otimes U'_s] U_E\,\ket{0^E}\ket{\phi}.
\end{equation}

\begin{prop}\label{propprot1}
If  $(U_E, U_F)$ is an attack on Protocol~1 such that $U_E$
 induces no error on $\TEST$ bits 
 then there are
states $\ket{E_i}$ in Eve's probe space such that for all $i\in\{0,1\}^N$,
\begin{equation}\label{defei1}
U_E\ket{0^E}\ket{i} = \ket{E_i}\ket{i}.
\end{equation}
If moreover  $ U_F$ induces no error on $\CTRL$ bits, then there are states
$\ket{F_{s,i}}$ of Eve's probe space such that for all $i\in \{0,1\}^N$, and all sequence $s$
of distinct elements of $[1\,..\,N]$,
\begin{equation}\label{deffs}
U_F\ket{E_i}\ket{i_s} = \ket{F_{s,i}}\ket{i_s}.
\end{equation}
\end{prop}
\begin{proof}
$U_E\ket{0^E}\ket{i}$ can be expanded as $\sum_j \ket{E_{ij}}\ket{j}$ and since for any $k$ there must
be a $0$ probability of getting $j_k$ different from $i_k$ 
(there is a non zero probability that Bob chooses bit $k$ as a $\TEST$ bit), $\ket{E_{ij}} = 0$ for
$j\neq i$ and thus (\ref{defei1}) holds with $\ket{E_i} = \ket{E_{ii}}$.
In Step~4, Bob publishes the
bit positions $s$ and, for Eve's attack to be unnoticeable by Alice, the state held by Alice
after $U_F$ is applied to $\ket{E_i}\ket{i_s}$ needs to be equal to $\ket{i_s}$. 
By Hilbert-Schmidt, this implies that the bipartite state $U_F\ket{E_i}\ket{i_s}$ must
be of the form $\ket{F}\ket{i_s}$. The pure state $\ket{F}$ depends here on
$i$, both through $\ket{E_i}$ and $i_s$, and also on the string $s$ chosen to 
select the reflected qubits, i.e.\ $\ket{F}$ is a function $i$ and $s$ and will be
written $\ket{F_{s,i}}$, giving Eq.~(\ref{deffs}).
\end{proof}

When the attack $(U_E,U_F)$ induces no error on $\TEST$ and $\CTRL$ bits then,
using (\ref{globalproto1}), (\ref{defei1}) and (\ref{deffs}),
\begin{equation}\label{summerize1}
[U_F\otimes I_B][ I_E\otimes U'_s] U_E\ket{0^E}\ket{i} = \ket{F_{s,i}}\ket{i_s}\ket{i_{\bar{s}}}.
\end{equation}

One can no longer expect Eve's final state  $\ket{F_{s,i}}$ after Alice sent state $\ket{i}$ and
Bob reflected the qubits specified by $s$ to be  constant, as 
is shown in the following example:
\begin{example}\label{example1}
 Let Eve's probe space be of dimension $N+1$ with basis states
$\ket{0}$ \ldots $\ket{N}$. Eve's initial state is $\ket{0}$.
 Let $U_E\ket{0}\ket{i} = \ket{|i|}\ket{i}$ and $U_F\ket{h}\ket{j} = 
\ket{h-|j|} \ket{j}$.
This means that $U_E$ puts in the probe  the Hamming weight $h=|i|$
of the string $i\in \{0,1\}^N$
if Alice sends state $\ket{i}$, and
$U_F$ subtracts from the probe the Hamming weight of the string $\ket{j}$
returned by Bob.
In particular $\ket{F_{s,i}} = \ket{|i| - |i_s|} = \ket{|i_{\bar{s}}|}$.  For $U_F$ to be
defined on all basis states assume the difference is modulo $N+1$.
Bob can clearly detect no error on $\TEST$ bits. Moreover, if
Alice sends $\ket{\phi} = \sum_i c_i \ket{i}$, the final state is
$\sum_i c_i \ket{|i_{\bar{s}}|}\ket{i_s}\ket{i_{\bar{s}}}$ and, once Bob has measured 
$\ket{i_{\bar{s}}}$, Eve's probe $\ket{|i_{\bar{s}}|}$ factors out and the resulting state in Alice's hands
is the same as if Eve had applied neither $U_E$ nor $U_F$, i.e.\ the final state had 
been $\sum_i c_i \ket{0}\ket{i_s}\ket{i_{\bar{s}}}$; no error can thus be detected on
$\CTRL$ bits.
\end{example}

Example~\ref{example1} shows  that Eve can  learn the Hamming weight $|i_{\bar{s}}|$ of
the string measured by Bob and stay completely invisible to Alice and Bob, i.e.\ 
induce no error on $\TEST$ and $\CTRL$ bits.
% This is why the choice of the $\INFO$ bits in Protocol~1 needs
% to be done carefully for the protocol to be nevertheless robust, i.e.\ so that Eve gets
% no information on the $\INFO$ string.
Therefore, in order to make protocol~1 robust, 
the choice of the $\INFO$ bits must be done in a more careful way.

But first, we need to show that
Eve can learn at most the Hamming weight of $i_{\bar{s}}$; this is a consequence
of Eq.~(\ref{dependshw}) below, which is derived from a sequence of lemmas. 
The first lemma states that all the bits in $i$ whose index are in $s$ can be flipped without
changing $\ket{F_{s,i}}$; in Protocol~2, this was true for all qubits in $i$,
but then, all the qubits were returned. In Protocol~1, only the qubits in $s$
are returned to Alice; the following lemma shows that for a fixed $s$, Eve's
state depends only on the bits kept by Bob.

\begin{lemma} \label{lemmaproperty2}
For any attack $(U_E, U_F)$ on Protocol~1 that induces no error on
$\TEST$ and $\CTRL$ bits, if $\ket{E_i}$ and $\ket{F_{s,i}}$ are given by
(\ref{defei1}) and (\ref{deffs})  then
\begin{align}
i_{\bar{s}} &= i'_{\bar{s}} \implies\quad\ket{F_{s,i}} = \ket{F_{s,i'}}. \label{property2}
\end{align}
\end{lemma}
\begin{proof}
The result is trivial if $s$ is empty. If not,
we follow the steps of the proof of Proposition~\ref{propprot2} and
prove this bitwise; let $k$ be an index in $s$, and
$i$ and $i'$ be such that $i_k=0$ and $i'_k=1$, all other bits being the same.
Assume wlg that $k$ is the first element of $s$ i.e.
$s=ks'$ and thus $i_s = i_k i_{s'}$.
If
Alice sends the state %$\ket{i_{[1..k-1]}}\ket{+}\ket{i_{[k+1..N]}}$ 
$\frac{1}{\sqrt{2}}[\ket{i}+\ket{i'}]$ 
i.e.\ the $k$th qubit sent
by Alice is $\ket{+}$ and all the other qubits % are those of $i$ and $i'$,
are prepared in the $Z$-basis, with bit values according to $i$,
then by linearity and Eq.~(\ref{summerize1})
the final state of the Eve+Alice+Bob system is
$
\frac{1}{\sqrt{2}} \Big[\ket{F_{s,i}}  \ket{0}+ \ket{F_{s,i'}}\ket{1}\Big]\ket{i_{s'}}\ket{i_{\bar{s}}}
$;
if we trace out all the qubits in $s'$ and $\bar{s}$ to keep only Eve's probe and
qubit $k$ in Alice's
hands, we get the state
\[
\frac{1}{\sqrt{2}}\Big [\ket{F_{s,i}}\ket{0} + \ket{F_{s,i'}}\ket{1}\Big ];
\]
writing $\ket{0}$ and $\ket{1}$  in terms of $\ket{+}$ and $\ket{-}$
and considering only those terms in the resulting state that
contain $\ket{-}$ gives
$
\frac{1}{2}\Big [ \ket{F_{s,i}} - \ket{F_{s,i'}}\Big ] \ket{-}
$;
and since the probability that Alice measures $\ket{-}$ as the $k$th qubit must be 0 (because $k\in s$), 
$\ket{F_{s,i}} - \ket{F_{s,i'}}=0$, i.e., $\ket{F_{s,i}} = \ket{F_{s,i'}}$. 
\end{proof}

The following lemma simply expresses the fact that, when Alice sends $\ket{i}$
and Bob reflects the qubits with indices in $s$ then Eve's final state depends only on
$\ket{i}$ and the state reflected by Bob.
\begin{lemma}\label{lemmaproperty1}
For any attack $(U_E, U_F)$ on Protocol~1 that induces no error on $\TEST$ and $\CTRL$ bits,
if $\ket{E_i}$ and $\ket{F_{s,i}}$ are given by
(\ref{defei1}) and (\ref{deffs})  then
for all $i$, $s$ and $s'$,
\begin{align}
i_s &= i_{s'}\implies\quad\ket{F_{s,i}} = \ket{F_{s',i}}.  \label{property1} 
\end{align}
\end{lemma}
\begin{proof}
If
$i_s = i_{s'}$ then $U_F\ket{E_i}\ket{i_s} = U_F\ket{E_i}\ket{i_{s'}}$ and 
thus $\ket{F_{s,i}}\ket{i_s} = \ket{F_{s',i}}\ket{i_{s'}}$. 
\end{proof}

When Eq.~(\ref{property1}) 
is used, we are using the fact that when Eve sees a qubit $\ket{0}$ (resp. a qubit $\ket{1}$)
coming back from Bob, then she cannot
tell to what qubit $\ket{0}$ (resp.  $\ket{1}$) sent by Alice this qubit corresponds provided
of course more than
one $\ket{0}$ (resp. $\ket{1}$) had been sent by Alice. 
The preceding lemmas can be used to show that, if Eve induces no error on $\TEST$ and
$\CTRL$ bits, then
Eve's intermediate state $\ket{E_i}$ just after $U_E$ is applied stays invariant when the bits in $i$
are permuted; let us first look at an example.

\begin{example}\label{example2}
Let $N = 4$ and $r=2$ and let us see that $\ket{E_{1011}} = \ket{E_{0111}}$
i.e.\ Eve's state after the attack $U_E$ on the qubits from Alice to Bob is the same whether
Alice sends state $\ket{1011}$ or $\ket{0111}$. By Eq.~(\ref{property2}), $\ket{F_{14,1011}}
= \ket{F_{14,0011}}$ which is Eve's final state when Bob reflects bits $1$ and $4$ and
Alice sends either $\ket{1011}$ or $\ket{0011}$. Similarly
$\ket{F_{24,0111}} = \ket{F_{24,0011}}$. We now use Eq.~(\ref{property1}) to 
get $\ket{F_{14,0011}} = \ket{F_{24,0011}}$ (Eve cannot tell if the returning $\ket{0}$ is
bit $1$ or bit $2$); those identities imply $\ket{F_{14,1011}} = \ket{F_{24,0111}}$.
We now go back to the definition of $F$; $\ket{F_{14,1011}}$ is Eve's final state if Alice sent
$\ket{1011}$ and Bob reflected the bits $1$ and $4$ and from Eq.~(\ref{deffs}) we get
$U_F\ket{E_{1011}} \ket{11} = \ket{F_{14,1011}}\ket{11}$ (bits $14$ being $11$).
Similarly $U_F\ket{E_{0111}}\ket{11} = \ket{F_{24,0111}}\ket{11}$ and since the r.h.s. members
are equal and $U_F$ is unitary,  $\ket{E_{0111}} = \ket{E_{1011}}$.
\end{example}

Following the lines of Example~\ref{example2}, we prove the following lemma.
\begin{lemma}\label{lem:Eiinvariant}
For any attack $(U_E, U_F)$ on Protocol~1 that induces no error on $\CTRL$ and $\TEST$ bits,
if $\ket{E_i}$ and $\ket{F_{s,i}}$ are given by
(\ref{defei1}) and (\ref{deffs})  then
 for all $i, i' \in \{0,1\}^N$
\begin{align}\label{Eiinvariant}
|i| = |i'|\quad\implies\quad \ket{E_i} &= \ket{E_{i'}}.
\end{align}
\end{lemma}
\begin{proof}
Eq.~(\ref{Eiinvariant}) means that $\ket{E_i}$ 
depends only on the number of ``$0$''s and ``$1$''s in $i$, 
not on their positions.
We need only show that any two (distinct) bits in $i$ can
be swapped without affecting $\ket{E_i}$ and, wlg,
$\ket{E_{01i''}} = \ket{E_{10i''}}$ for any
$i'' \in \{0,1\}^{N-2}$\footnote{If $n \geq 1$ then 
$N=\lceil 8n(1+\delta)\rceil \geq 8$ and $N-2 \geq 1$.}. 
Let $s'$ be any sequence of distinct
elements of $[3\, ..\, N]$;  
$\ket{F_{1s',10i''}} = \ket{F_{1s',00i''}}$ and
$\ket{F_{2s',01i''}} = \ket{F_{2s',00i''}}$ by Eq.~(\ref{property2}); also
$\ket{F_{1s',00i"}} = \ket{F_{2s',00i''}}$ by Eq.~(\ref{property1}) and thus
% \begin{equation}\label{equaleffs}
$\ket{F_{1s',10i''}} = \ket{F_{2s',01i''}}.$
%\end{equation}
Using Eq.~(\ref{deffs}),
\begin{align*}
U_F \ket{E_{10i''}} \ket{1i_{s'}} &=\ket{F_{1s',10i''}}\ket{1i_{s'}}   &&\text{($i=10i''$; $s=1s'$)}\\
U_F \ket{E_{01i''}}\ket{1i_{s'}} &= \ket{F_{2s',01i''}}\ket{1i_{s'}}   &&\text{($i=01i''$; $s=2s'$)}
\end{align*}
and, since  $i_{s'}$ is the same for $i=01i''$ and $i=10i''$ and $\ket{F_{1s',10i''}} = \ket{F_{2s',01i''}}$
the r.h.s. are equal and so
 $\ket{E_{10i''}} = \ket{E_{01i''}}$.
\end{proof}

\begin{example}\label{example3}
Lemma~\ref{lemmaproperty2} allows replacing all bits indexed by $s$ by $0$ without changing $\ket{F_{s,i}}$.
This means for example that if $i=1010$ and $s=34$, then
$\ket{F_{s,i}} = \ket{F_{34,1010}} = \ket{F_{34,1000}}$; similarly if $i'=0101$ and $s'=12$,
then $\ket{F_{s',i'}} =
\ket{F_{12,0101}} = \ket{F_{12,0001}}$. This means that $\ket{F_{s,i}}$ depends only
on the bits not indexed by $s$, i.e.\ the bits indexed by $\bar{s}$. Here $i_{\bar{s}} = 10$ and
$i'_{\bar{s}'} = 01$; those two strings have the same Hamming weight. Let us see that
they give the same final state for Eve.  By Eq.~(\ref{deffs})
$U_F\ket{E_{0001}}\ket{00} = \ket{F_{12,0001}}\ket{00}$ (Bob reflects bits $12$) and
$U_F\ket{E_{1000}}\ket{00} = \ket{F_{34,1000}}\ket{00}$ (Bob reflects bits $34$).
We know from Lemma~\ref{lem:Eiinvariant} that
$\ket{E_{0001}} = \ket{E_{1000}}$;
this implies $\ket{F_{12,0001}} = \ket{F_{34,1000}}$ and thus
$\ket{F_{s,i}} = \ket{F_{s',i'}}$. 
\end{example}

Example~\ref{example3} provides the intuition behind the proof of the
next proposition that is for Protocol~1 what Proposition~\ref{propprot2} is for Protocol~2.

\begin{prop}%\label{propprot1}
If $(U_E, U_F)$ is an attack on Protocol~1 that induces no error on $\TEST$ and $\CTRL$ bits, 
and if $\ket{E_i}$ and $\ket{F_{s,i}}$ are given by
(\ref{defei1}) and (\ref{deffs})  then
for all $s$ and $s'$ of the same length $r \geq 0$, and all $i,i' \in \{0,1\}^N$,
\begin{align}
|i_{\bar{s}}| = |i'_{\bar{s}'}| \quad \implies\quad \ket{F_{s,i}} &= \ket{F_{s',i'}}.\label{dependshw}
\end{align}
\end{prop}
\begin{proof}
Let $j$ and $j'$ be defined by
$j_s = j'_{s'} = 0^r$, $j_{\bar{s}} =i_{\bar{s}}$ and $j'_{\bar{s}'} = i'_{\bar{s}'}$. Then
$\ket{F_{s,i}} = \ket{F_{s,j}}$ and $\ket{F_{s',i'}} = \ket{F_{s',j'}}$ by Eq.~(\ref{property2}).
Since $|j'| = |j|$, 
 $\ket{E_j} = \ket{E_{j'}}$ by Eq.~(\ref{Eiinvariant}); by  Eq.~(\ref{deffs}), 
$U_F\ket{E_j}\ket{j_s} = \ket{F_{s,j}}\ket{j_s}$ and $U_F\ket{E_{j'}}\ket{j'_{s'}} = \ket{F_{s',j'}}\ket{j'_{s'}}$
and thus, since $\ket{j_s} = \ket{j'_{s'}} = \ket{0^r}$,  $\ket{F_{s,j}} = \ket{F_{s',j'}}$.
\end{proof}

Eq.~(\ref{dependshw}) can be rewritten 
$\ket{F_{s,i}} = \ket{F_{|i_{\bar{s}}|}}$ 
%% RG8b - No need to define |F_h>. [Original parag]:
representing Eve's final state 
when the Hamming weight of the string measured by Bob is $|i_{\bar{s}}|$.

\begin{theorem}\label{theoreminfoprotocol1}
With any attack % $(U_E,U_F)$ 
on Protocol~1 that induces no error on $\TEST$ and $\CTRL$ bits,
the eavesdropper can learn at most the number of ``$0$''s and ``$1$''s measured by Bob
and Eve's final state can be written $\ket{F_{|i_{\bar{s}}|}}$.
% Eve's final state  may depend
% only on the Hamming weight of the string $i_{\bar{s}}$.
\end{theorem}
\begin{proof}
Let $(U_E, U_F)$ be an arbitrary attack that induces no error on $\TEST$ and $\CTRL$  bits.
If Alice sent any superposition $\ket{\phi} = \sum_{i \in \{0,1\}^N} c_i \ket{i}$ and Bob returned
the bits selected by $s$, then using linearity and Eq.~(\ref{summerize1}) 
with $\ket{F_{s,i}}= \ket{F_{|i_{\bar{s}}|}}$ for all $i$ gives
\begin{equation}\label{finalstateprot1}
\sum_i c_i  \ket{F_{|i_{\bar{s}}|}} \ket{i_s}\ket{i_{\bar{s}}}
\end{equation}
as the state describing the final Eve+Alice+Bob system. Once Bob measures $\ket{i_{\bar{s}}}$
the state is projected onto a state where $\ket{F_{|i_{\bar{s}}|}}$ factors out and Eve is
thus left with a state that depends only on
the Hamming weight $|i_{\bar{s}}|$ of the string measured by Bob
and thus can learn
at most that Hamming weight. Since she knows the length of $s$, this means she can learn
at most the number of  ``$0$''s and ``$1$''s measured by Bob.
\end{proof}

\subsubsection{Information leaked by Protocol 1.}
As shown in Example~\ref{example1},
 Eve can indeed learn the Hamming weight of the string measured by Bob.
This is % the reason 
why the mock protocol of %the \textit{basic idea} paragraph
Section~\ref{sec:idea}
failed. There was only one $\SIFT$ bit and
no permutation could ever hide its value.

%The state
%$\ket{\phi}$ sent by Alice is always a product of $Z$ basis states $\ket{0}$, $\ket{1}$ or of
%$X$-basis states $\ket{+} = \frac{1}{\sqrt{2}}[\ket{0}+\ket{1}]$, $\ket{-} = \frac{1}{\sqrt{2}}[
%\ket{0} -\ket{1}]$. Its computational-basis coefficients
%$c_i$ are equal to $\pm \left(\frac{1}{\sqrt{2}}\right)^{n_X}$, where $n_X$ is the
%number of qubits sent by Alice in the $X$ basis; 
%the actual values of $c_i$ are nowhere needed in the proofs.
{}From Eq.~(\ref{finalstateprot1}) one also sees
that the probability of Bob measuring $i_{\bar{s}}$ is unaffected by Eve's attack, just
because the norm of $\ket{F_{|i_{\bar{s}}|}}$ is $1$ (this is a normalized state);
Eve's attack has no effect at all on Bob's statistics. 
The $\SIFT$ bits are equal to the random $Z$-bits chosen by Alice; 
the $X$ bits measured by Bob
are also random bits, as they would be without Eve's attack.  

 From the string of $N-r$ bits (whose indices are in $\bar{s}$) 
measured by Bob, about half the bits  are discarded because Alice sent the  
corresponding qubit in the $X$-basis. The bits left are the $\SIFT$ bits; $n$ of
them are used as $\TEST$ bits, the others serve
as a pool selecting the $\INFO$ bits. Eve's knowledge of $|i_{\bar{s}}|$
provides indirect knowledge on the statistics of occurrence of
``$0$''s and ``$1$''s in the $\INFO$ bits and the protocol would nevertheless not be robust if the
$\INFO$ string was obtained by picking randomly 
 $n$ bits from the $\SIFT$ bits not used as $\TEST$ bits (or the first $n$ ones available
 as in Protocol~2). We now give an asymptotic bound on Eve's accessible information.
%We now show that Eve's information on the $\INFO$ string if it is not chosen cleverly is
%at most $0.293$ bits, except with negligible probability; indeed, the probability that
%Bob chooses a string $s$ such that $\bar{s}$ is of length less than $4n$ is exponentially small.
%\ran{Is this sentence appropriate here? maybe transferring after the proof, or saying that
%the proof requires $N-r>4n$.}

\begin{theorem}\label{thprotone}
For any attack on Protocol~1 that induces no error on $\TEST$ and $\CTRL$ bits,
Eve's information on the $\INFO$ string is
asymptotically less than $0.293 + O(n^{-1})$ bits. %  if ${N-r} \geq 4n$.
\end{theorem}
\begin{proof}
Let $N-r=kn$ be the number of bits measured by Bob; it is expected that $k=4(1+\delta)$;
those bits are all random but 
Eve knows their Hamming weight. Also known are the indices of the $\SIFT$ bits,
of the $\INFO$ bits, as well as the indices and values of the $\TEST$ bits.
Eve thus knows the Hamming weight $W$ of the $kn-n$ remaining 
random bits that are not $\TEST$;
$W$ is distributed binomially, with $kn-n$ trials and probability $1/2$ of success.
The entropy of a binomial distribution with $n$ trials and probability $p$ of
success is $1/2 \lg(2\pi e p(1-p)n) + O(1/n)$ 
where $O(1/n)$ is the error~\cite{JS99}\footnote{When
$n$ is large, the binomial $B(n,p)$ is well approximated by a normal with variance $\sigma^2=np(1-p)$,
whose entropy  is $\lg(\sigma \sqrt{2\pi e}) = \lg{\sqrt{2\pi e p(1-p)n}}$. With a factor of
$\frac{1}{\lg(e)}$ this rewrites
$\frac{1}{2}\log n + \frac{1}{2} + \log \sqrt{2\pi  (1-p)}$ as in \cite{JS99} where is proven a result implying the
error is of order $\frac{1}{n}$. A simple computer program shows that for
$p=0.5$ and $n=20$ the error is already less than $3.4\times 10^{-4}$.}; the
entropy
$H(W\mid k)$ is thus
\begin{equation}
H(W\mid k) =  \frac{1}{2} \lg\left(\frac{1}{2}\pi e (k-1)n \right) + O\left(\frac{1}{n}\right).
\end{equation}
For any particular $n$-bit $\INFO$ string $x$, the entropy of $W$ given $x$ and $k$
is the entropy of the a binomial distribution with $kn-2n$ trials (for the $kn-2n$ remaining random bits)
and is thus
\begin{equation}
H(W\mid x, k) = \frac{1}{2} \lg\left(\frac{1}{2}\pi e (k-2)n \right) + O\left(\frac{1}{n}\right).
\end{equation}
The bits of the $\INFO$ string are random bits chosen by Alice
and  the strings $x$ are thus equally likely; this implies $H(W\mid X, k) =
 H(W\mid x, k)$.
The information Eve gains on $X$ when $W$ is known is, for any fixed $k$, % by definition
 $H(X \mid k) - H(X\mid W, k)$. 
It is a basic fact from information theory that $H(X\mid k) - H(X\mid W, k) = 
H(W\mid k) - H(W\mid X, k)$
and Eve's information is thus
\begin{align} \nonumber
I(W;X, k) &= H(W\mid k) - H(W \mid X, k) \\ \nonumber
 &=  \frac{1}{2} \lg{\frac{k-1}{k-2}} + O\left(\frac{1}{n}\right)\\ 
 &= \frac{1}{2}\lg\left( 1+ \frac{1}{k-2}\right) + O\left(\frac{1}{n}\right).
\end{align}
For $k\geq 4$, $I(W;  X, k) \leq 0.293 + O(n^{-1})$; the probability that
$k < 4$ is exponentially small in $n$ and thus 
$I(W; X) \leq \sum_k I(W; X, k) p(k) < 0.293 + O(n^{-1})$. 
%%% MB13c 
% Details: 
% $
% I(W; X) = H(W) - H(W\mid X) \leq H(W) - H(W\mid X, K) = I(W; X, K)%$; $I(W;X,K) 
% = 
% \sum_{k\geq 4} I(W; X,k)p(k) + \sum_{k< 4} I(W; X,k)p(k) < 0.293 + O(n^{-1}) + \lceil 8n(1+\delta)\rceil P[K<4]$
% and $O(n^{-1})$ wins. 

\end{proof}

% With $k \geq 4$ this gives less than $0.293 + O(n^{-1})$ bits of information. 
% The probability that $k < 4$ is exponentially small; even if Eve then gets full information,
% the expected number of bits she gets is of order $0.293 + O(n^{-1})$. 
%I will not detail the rest
% \MB{This also holds if we restrict to the case the protocol succeeds, since it succeeds with
% probability $\geq 1/2$ when $n$ is large.}
%\ran{The intention of the last sentence (from ``even if Eve'') is not clear.
%What did you actually mean?}

%% FROM NEWONETHIRD

\subsection{Properties of Protocol~1$'$.}
\subsubsection{The information contained in the $\INFO$ string.}
Alice chooses randomly $y\in I_{n,\epsilon}$
to send as the $\INFO$ string. The information contained in $y$ is thus the entropy
of a uniform distribution on $I_{n,\epsilon}$.
\begin{prop} If $\epsilon > 0$, the entropy of the uniform distribution on $I_{n,\epsilon}$ is
exponentially close to $n$ (its distance to $n$ is of order $e^{-\Omega(n)}$).
\end{prop}
\begin{proof}
For any integer $N>0$, the entropy of the uniform distribution on a set of $N$ elements is  $\log_2(N)$. We are thus
looking for a lower bound on $\log_2(|I_{n,\epsilon}|)$.

Let  $Y=(Y_1,\ldots, Y_n)$ be a uniformly distributed random variable on $\{0,1\}^n$;
the $Y_i$ are independent Bernoullis with probability $p=1/2$.
 Let $\overline{Y} = \sum_{i=1}^n Y_i/n$; $\overline{Y}$ is nothing but $|Y|/n$;
the expectancy $E[\overline{Y}]$ is $1/2$. 
\begin{align*}
n-\log_2(|I_{n,\epsilon}|) &= -\log_2\left(\frac{|I_{n,\epsilon}|}{2^n}\right) \\
 &= -\log_2\left(\Pr\left[ \left|\overline{Y} - \frac{1}{2}\right| \leq \frac{\epsilon}{2}\right]\right).
% \\
%  &= -\log_2\left(1 - \Pr\left[ \left|\overline{Y} - \frac{1}{2}\right| > \frac{\epsilon}{2}\right]\right)
\end{align*}
%By symmetry, \MB{$$\Pr \left[\left|\overline{Y} - \frac{1}{2}\right| > 
%\frac{\epsilon}{2} \right] = 2\cdot\Pr \left[\overline{Y} - \frac{1}{2} > \frac{\epsilon}{2}\right], $$}
%and using Hoeffding's inequality (Theorem~\ref{hoeffding})
%\[\displaystyle \Pr\left[\overline{Y}-\frac{1}{2} > 
%\frac{\epsilon}{2}\right] \leq \exp\left(-\frac{\epsilon^2}{2}n\right),\text{\quad we get}\]
% and thus
By Hoeffding's inequality \eqref{ValAbsHoeffding}
\[\displaystyle \Pr\left[\left|\overline{Y}-\frac{1}{2}\right| > 
\frac{\epsilon}{2}\right] \leq 2\exp\left(-\frac{\epsilon^2}{2}n\right),\text{\quad and thus}\]
\begin{align*}
n-\log_2(|I_{n,\epsilon}|) &\leq -\frac{1}{\ln(2)} \ln\left(1-2\exp\left(-\frac{\epsilon^2}{2}n\right)\right).
\end{align*}
For $0<x<0.5$, it is easy to verify that
$
-\ln(1-x) \leq 3x/2
$;
thus, for  $n$ large enough
(e.g.\ $ n > \ln(16)/\epsilon^2$),
\[ %\displaystyle
n-\log_2(|I_{n,\epsilon}|) \leq \frac{3}{\ln{2}} \exp\left(-\frac{\epsilon^2}{2}n\right). \qedhere
\] 
\end{proof}

%some introduction for this proposition was missing..
While the entropy is $\approx\!{n}$ when $\epsilon > 0$, 
we now show that it has a gap of $0.5\log_2(n)$ bits when $\epsilon =0$.
\begin{prop}
For $\epsilon = 0$, the entropy of the uniform distribution on $I_{n,0}$ is asymptotically
equal to ${n - 0.5\log_2(n)} -0.5(\log_2(\pi)-1)$
\end{prop}
\begin{proof}
Stirling's formula gives
$$ \displaystyle
\lim_{n\to \infty} \binom{n}{n/2}\Big/\frac{2^n}{\sqrt{\pi n/2}} = 1.
$$ 
We get the result by taking the $\log$.
\end{proof}

Thus, by choosing $\epsilon > 0$, we avoid asymptotically 
loosing more than $0.5\log_2(n)$ bits of information.

\subsubsection{Probability of aborting Protocol 1$'$.}
The protocol aborts if there are 
less than $h$ zeros or $h$
ones left in the $\SIFT$ string after $n$ $\TEST$ bits have been chosen, where
$h=\lfloor (1+\epsilon)n/2\rfloor$. We prove that this
occurs with a probability that decreases exponentially with $n$.

\begin{prop} For any $0 \leq \epsilon < \delta$ and $\epsilon \leq 1$
fixed by the protocol, the probability that it aborts 
is exponentially small.
\end{prop}
\begin{proof} 
We begin with showing that, besides an exponentially small probability,
the number of $\SIFT$ bits is larger than $N/4$. 
We follow  by showing that this is enough for having at least $h$ zeros and ones,
except for exponential probability.
Let $\delta'$ be a real number such that
$\epsilon < \delta' < \delta$.  Let $N=\lceil 8n(1+\delta)\rceil$.
For $i$ such that $1 \leq i \leq N$, let $X_i=1$ if the qubit $i$ is $\SIFT$ and $X_i=0$ otherwise.
The variables $X_i$ are clearly independent; their distribution is a Bernoulli with $p=0.25$, 
%% RG12d - refering to  FIGURE 2, instead of talking about p=0.25.
% indeed
% Alice sends randomly  a qubit in $\{\ket{0}, \ket{1}, \ket{+}, \ket{-}\}$ and
% Bob randomly  measures or reflects; the probability of Bob measuring a qubit Alice sent as
% $\ket{0}$
% or $\ket{1}$ is thus $0.25$. 
as shown in Fig~\ref{fig:protsum}.
The random variable $S$ giving the number of $\SIFT$ bits
is $S = \sum_{i=1}^N X_i$. 
%% RG12d \Bar truned into \overline
Denote $\overline{X} = S/N$; it is clear that $E[\overline{X}] = 1/4$, and 
we can bound ${\Pr[S \leq N/4]}$ using Hoeffding (Theorem~\ref{hoeffding}),
% Let $\overline{X}=S/N$ and $N \geq 8n(1+\delta)$
 \begin{align*}
\Pr\big[S \leq 2n(1+\delta')\big] &\leq \Pr\left[\overline{X} \leq \frac{1}{4}\frac{1+\delta'}{1+\delta}\right] \\
 &\leq \Pr\left[ \overline{X} -\frac{1}{4} \leq - \frac{\delta - \delta'}{4(1+\delta)}\right] \\
  & \leq \exp\left({-\frac{1}{8}\left( \frac{\delta-\delta'}{1+\delta}\right)^2n}\right)
\end{align*}
% and by Hoeffding's inequality (Theorem~\ref{hoeffding})
% \[
% \Pr\big[S \leq 2n(1+\delta')\big] \leq \exp\left({-\frac{1}{8}\left(
% \frac{\delta-\delta'}{1+\delta}\right)^2n}\right)
% \]
and thus
\[
\Pr\big[S > 2n(1+\delta')\big] \geq 1 - e^{-k_1n}
\]
for $k_1 = 1/8\left[(\delta' - \delta)/(1+\delta)\right]^2$.

For each $S> 2n(1+\delta')$, the $S$ bits are distributed uniformly. After $n$ $\TEST$ bits are chosen,
the remaining $S - n > 2n(1+\delta') - n =  n(1+2\delta')$ bits are still uniformly distributed.
Every time there are at least $h$ zeros and $h$ ones after
the $n$ $\TEST$ bits are chosen, and %in particular when 
in addition there are more than
$2n(1+\delta')$ $\SIFT$ bits, the protocol succeeds. 
As a consequence, the probability of success is larger than or equal to
the probability that $S > 2n(1+\delta')$ times the probability that the $S-n$ remaining bits
contain at least $h$ zeros and $h$ ones, given that $S > 2n(1+\delta')$.
Let $V$ be the length of the string $v$, i.e.\ 
 $V = S-n > n(1+2\delta')$. Let us index the bits in $v$ from 1 to $V$, let
 $Z_i = 1$ if bit $i$ is $0$ and $Z_i = 0$ otherwise, let $Z=\sum_{i=1}^V Z_i$  and let $\overline{Z} 
= Z/V$; $Z$ is thus the number of bits equal to $0$ in $v$; the $Z_i$ are Bernouilli with $p=1/2$ and
are independent.
Let us denote $\Pr_V$ the probability conditional to that particular value of $V$. The probability
that there are strictly less than $h$ zeros in $v$ is bounded by
\begin{align*}
\Pr_V[Z < h] &\leq \Pr_V[Z \leq (1+\epsilon)n/2] \\
 &= \Pr_{V}[\overline{Z} \leq (1+\epsilon)n/(2V)] \\
 &\leq \Pr\left[\overline{Z} \leq \frac{1+\epsilon}{2(1+2\delta')}\right] \\
 &= \Pr\left[\overline{Z} - \frac{1}{2} \leq -\frac{2\delta' - \epsilon}{2(1+2\delta')}\right]
\end{align*}
where $\delta' > \epsilon$ by hypothesis
and again, by Hoeffding (Theorem~\ref{hoeffding}), 
the probability that there are not enough 
zeros when $S>2n(1+\delta')$ is  bounded by
\[
\exp\left( - \frac{1}{2}\left(\frac{2\delta' - \epsilon}{1+2\delta'}\right)^2 n\right)
\]
and the probability that there are at least $h$ zeros and $h$ ones when
$S > 2n(1+\delta')$ is larger than or equal to
\(
{1 - 2e^{-k_2n}}
\) 
with $k_2 = \displaystyle\frac{1}{2} \left(\frac{2\delta' - \epsilon}{1+2\delta'}\right)^2$.
As a consequence, the probabilty that the protocol succeeds is at least
\[
(1 - e^{-k_1n})(1 -2e^{-k_2n}) = 1 - e^{-k_1n} -2e^{-k_2n} + 2e^{-(k_1+k_2)n}
\]
which is more that
\(
1 - 3e^{-kn}
\) 
with $k = \min\{k_1,k_2\}$. It is exponentially close to 1 with $n$.
\end{proof}

\subsection{Complete robustness of Protocol 1$'$.}

The assumption is that Eve's attack is undetectable,
and we want to show that she gets no information
on the $\INFO$ string. During the execution of the protocol,
Eve learns which are the $\TEST$ bits,
she learns their values, she learns the number
of bits measured by Bob and, more importantly, 
her attack allows her to know their Hamming weight.
We group all those data in the multivariate random variable 
$\mathbf{R}$ of which the 
details will be irrelevant; $\mathbf{r}$ will be a particular set of data. The execution of the protocol
also gives Eve the set of indices $q$ such that $v_q = y$. What we want to show is that
\begin{equation}\label{whattoprove}
I(Y; Q, \mathbf{R}) = 0,
\end{equation}
i.e.\ the mutual information between the $\INFO$ string $y$ and what Eve knows, namely
 $(q, \mathbf{r})$, is zero.

\subsubsection{Probabilistic setup.}

Let $F$ be the set of indices measured by Bob.
By Theorem~\ref{theoreminfoprotocol1}, if Eve is unnoticeable, her final state may depend only
on $|i_F|$.
% \RG{\sout{where $F=\bar{s}$ is the set of indices measured by Bob.}}
Eve's final state does not depend on $y$ either.
That implies that, whatever $\mathbf{r}$ Eve learns and for any value $y\in \{0,1\}^n$
\begin{equation}\label{eqsetup}
|i_F| = |i'_F| \implies p(i \mid y, \mathbf{r}) = p(i' \mid y, \mathbf{r}).
\end{equation}
% ; of course $\mathbf{r}$ contains test bits in $i$. 
%That also implies that  for any $E\subseteq F$, $p(i_E \mid \mathbf{r}) = p(j_E \mid \mathbf{r})$
%as soon as if $|i_F| = |j_F|$.

For $h = \lfloor (1+\epsilon)n/2\rfloor$, 
Alice chooses $2h$ indices in $F$ that are $\SIFT$ %$\INFO$ 
bits and not $\TEST$ bits, say $E$.
Let $E_h$ be the set of all balanced  strings $x$ indexed
by $E$, 
i.e.  
\begin{equation}
E_h = \{ x \in \{0,1\}^E \mid |x| = h\}.
\end{equation}
% }
\begin{lemma}\label{lemprobx}
For any $x, x' \in E_h$, 
\begin{equation}\label{eqpx}
p(x\mid y, \mathbf{r}) = p(x' \mid y, \mathbf{r})= \frac{1}{|E_h|} = \frac{{h!}^2}{(2h)!}.
\end{equation}
\end{lemma}
\begin{proof}
To simplify notations, and without loss of generality, assume that $E= \{1,\ldots, 2h\}$
so that $\{0,1\}^E$ is the set of bitstrings with indices from $1$ to $2h$,
and $F = \{1,\ldots, |F|\}$;
$p(x \mid y, \mathbf{r}) = \sum_{vv'} p(xvv' \mid y, \mathbf{r})$ where 
$v$ are all bitstrings
with indices in $\{2h+1,\ldots, |F|\}$ and $v'$ are those with indices in $\{|F|+1,\ldots, N\}$;
similarly $p(x' \mid y, \mathbf{r}) = \sum_{vv'} p(x'vv'\mid y, \mathbf{r})$;
if we let $i=xvv'$ and $i'=x'vv'$ then $xv= i_F$,
$x'v = i'_F$ and $|i_F| = |xv| = |x|+|v| = |x'|+|v| = |x'v| = |i'_F|$ and thus,
by (\ref{eqsetup}), $p(i\mid y,\mathbf{r}) = p(i'\mid y,\mathbf{r})$ and the
two sums are equal.
\end{proof}

\subsubsection{Combinatorial lemmas.}

Given a set $E$ and $k \leq |E|$, we denote $\mathcal{P}(E,k)$ the set of permutations of
$k$ elements in $E$, i.e.\ the set of strings
 $q_1\ldots q_k$  of $k$ distinct elements in $E$\footnote{For simplifying the notations,
bits keep their indices even when they appear in substrings.};
% If we let $2h = |E|$, 
% We need to apply this with E_0 and E_1 in a lemma; we cannot fix to 2h
\begin{equation}\label{eq1}
\left| \mathcal{P}(E,k)\right| =  \frac{|E|!}{(|E|-k)!}.
\end{equation}
{}From now on, $\epsilon$ such that $0\leq \epsilon \leq 1$, $\epsilon < \delta$ will be fixed, as
well as $h = \lfloor (1+\epsilon)n/2\rfloor$ and $E$, a set of $2h$ indices of $\SIFT$ bits that are
not $\TEST$ bits.
For $y\in I_{n,\epsilon}$ and $x\in E_h$ we let
\[
Q(x,y) = \left\{q \in \mathcal{P}(E,n) \mid x_q = y\right\}.
\]

\begin{lemma} % If $\epsilon \geq 0$ and $h=\lfloor (1+\epsilon)n/2\rfloor$,  then
For all $y\in I_{n,\epsilon}$ and $x\in E_h$
the number of elements $\left| Q(x,y) \right|$ of \ $Q(x,y)$ is
\begin{equation}\label{eqqxy}
\left| Q(x,y) \right| = \frac{{h!}^2}{(h-n+|y|)!\times (h - |y|)!}
\end{equation}
\end{lemma}
\begin{proof} 
A string $y\in \{0,1\}^n$ is in $I_{n,\epsilon}$ 
if and only if it contains at most $h$ zeros and $h$ ones.
Let $E_0 = \{j \in E \mid x_{j} = 0\}$ and $E_1 = \{j \in E  \mid x_{j} = 1\}$;
% $|E_0| = h$ and $|E_1| = h$
$|E_0| = |E_1| = h$ 
and the permutations $q$ such that $x_q = y$ 
are in 1---1 correspondence with the elements of
\[
\mathcal{P}\left(E_0, n-|y|\right) \times \mathcal{P}\left(E_1, |y|\right)
\]
corresponding to the $n-|y|$ indices giving a $0$ in $y$ 
and the $|y|$ indices giving a $1$ in $y$.
The result follows from~(\ref{eq1}).
\end{proof}

\begin{lemma} % If $\epsilon \geq 0$, $h=\lfloor (1+\epsilon)n/2\rfloor$,
 For all $q\in \mathcal{P}(E, n)$ and $y\in I_{n,\epsilon}$ % then
\begin{equation}\label{eqnbxstqinQxy}
 \left| \left\{ x \in E_h \mid q\in Q(x,y) \right\}\right| 
 = \binom{2h-n}{h-|y|} 
 \end{equation}
\end{lemma}
\begin{proof}
A string $x\in E_h$ is such that $q\in Q(x,y)$ 
if and only if it satisfies $x_q = y$; this means that 
$x_{q_1} = y_1$, \ldots, $x_{q_n} = y_n$ 
(bits indexed by $q$ are fixed), the other bits are arbitrary
provided there is a total of $h$ bits equal to $0$ and 
$h$ bits equal to 1; the desired strings are 
thus obtained by
filling  the $2h-n$ bit positions
whose indices are not in the list $q$ with
$h-|y|$ bits equal to $1$ (and the others equal to $0$);  
 there are
$\displaystyle\binom{2h -n}{h - |y|} $ such strings.
\end{proof}
Eq.~(\ref{eqnbxstqinQxy}) can be rewritten
\begin{equation}\label{eqnbxstqinQxyb}
 \left| \left\{ x \in E_h \mid x_q = y \right\}\right| =
 \frac{(2h-n)!}{(h-n+|y|)! (h-|y|)!}.
\end{equation}
\subsubsection{Proof of robustness.}
We want to show that $q$ leaks no information on $y\in I_{n,\epsilon}$. 
For any fixed $x\in E_h$ and $y
\in I_{n,\epsilon}$, % with $h={\lfloor (1+\epsilon)n/2\rfloor}$,
the probability that Alice sends $q$  is 
$1/{|Q(x,y)|}$ if $q\in Q(x,y)$, 0 otherwise, 
independently of any value of $\mathbf{r}$:% Ran: Why not deleting this line too?
\begin{equation}\label{eqpqmidxy}
p(q \mid x, y,\mathbf{r}) = \begin{cases}
  \displaystyle\frac{1}{|Q(x,y)|}   
  &\text{if $x_q = y$}\\
  0 &\text{otherwise}
\end{cases}.
\end{equation}

\begin{lemma}\label{labprop} % If $\epsilon \geq 0$, 
% $h=\lfloor (1+\epsilon)n/2\rfloor$, then for all
For all 
values of $\mathbf{r}$, all
$y\in I_{n,\epsilon}$ and all $q \in \mathcal{P}(E,n)$
\begin{equation}\label{eqindep}
p(q\mid y,\mathbf{r}) =  \frac{(2h- n)!}{(2h)!}.
\end{equation}
\end{lemma}
\begin{proof}
\begin{align*}
p(q\mid y,\mathbf{r}) &= 
     \sum_{x\in E_h} p(q\mid x,y,\mathbf{r})p(x\mid y,\mathbf{r}) \\
  &= \sum_{x\in E_h \mid x_q = y} \frac{1}{|Q(x,y)|} 
      \frac{h!^2}{(2h)!} 
  % &&   \text{by (\ref{eqpx}, \ref{eqpqmidxy})}
\\
  &= \sum_{x\in E_h \mid x_q = y} \hspace*{-1em} 
   \frac {(h-n+|y|)!(h - |y|)!}{(2h)!} 
  %&& \text{by (\ref{eqqxy})}
\\
  &= \frac{(2h-n)!(h-n+|y|)!(h -|y|)!}
          {(h-n+|y|)!(h-|y|)!(2h)!} 
 %&& \text{by (\ref{eqnbxstqinQxyb})}
\\
  &= \frac{(2h - n)!}{(2h)!} %\qedhere
\end{align*}
where the second equality is due to 
\eqref{eqpx} and \eqref{eqpqmidxy}, and the third and forth equalities
are given by  \eqref{eqqxy} and \eqref{eqnbxstqinQxyb}.
\end{proof}

\begin{theorem}For all 
$\epsilon$ and $\delta$ such that
$0  \leq \epsilon \leq 1$ and $\epsilon < \delta$, 
the protocol is completely robust, i.e.\ if Eve is
undetectable by the legitimate parties, then
$I(Y; Q, \mathbf{R}) = 0$. 
\end{theorem}
\begin{proof}
The parameters $n$ and $\epsilon$ are constants of the protocol; they are fixed before all
random choices of Alice or Bob, and all measurements. So is the value
$h = {\lfloor (1+\epsilon)n/2\rfloor}$. 
The right-hand side of Eq.~(\ref{eqindep}) is thus
a constant\footnote{From (\ref{eq1}), we see that 
$P(q\mid y,\mathbf{r}) = 1/|\mathcal{P}(E,n)|$
which is the probability of a random $n$-permutation of $|E| = 2h$ elements.} and
Lemma~\ref{labprop}  implies that  the random
variables $Q$ and $(Y,\mathbf{R})$ are independent:
 $p(q, y, \mathbf{r}) = p(q)p(y,\mathbf{r})$;
the variables $Y$ and $\mathbf{R}$ must also be independent,
because Alice chooses $y$
randomly, independently of everything else:  $p(y,\mathbf{r}) = p(y)p(\mathbf{r})$.
This implies that
$p(q,y,\mathbf{r}) = p(q)p(y)p(\mathbf{r})$,
$Y$ is independent of $(Q,\mathbf{R})$, therefore $I(Y; Q,\mathbf{R} ) = 0$.
\end{proof}

% \enlargethispage{3\baselineskip}
\section{Conclusion}
We presented two protocols for QKD with one party 
performing only classical
operations: measure a qubit in the classical $\{0,1\}$ basis, 
let the qubit pass undisturbed back to its sender, 
randomize the order of several qubits, or resend a qubit after
its measurement.
We proved the robustness of these protocols;
this provides intuition why we believe they are secure.
We hope that this work sheds light on ``how much quantumness'' is required 
in order to perform the classically-impossible task of secret key
distribution.
This work extends the previous work~\cite{Ken07} 
and the conference version~\cite{BKM07conf}
by two aspects: it proves robustness of the measure-resend SQKD Protocol for a more general scenario
and proves the full robustness of a randomization-based SQKD Protocol, eliminating any information leak to the
adversary. 

Note that in this work we assumed perfect qubits.
We leave the examination 
of our protocol against PNS and other implementation-dependent attacks
to future research.
This work was partially supported by the Israeli MOD.
We thank Moshe Nazarathy for providing the motivation for this research.

\section*{APPENDIX}
\begin{theorem}[Hoeffding]\label{hoeffding}
(Hoeffding~\cite{Hoeffding63}) If \ $X_1$, \ldots, $X_n$ are independent 
random variables with finite first and second moments,  
$\Pr[a_i \leq X_i \leq b_i] = 1$ for $1\leq i \leq n$,
and  $\displaystyle\bar{X} = \frac{1}{n}\sum_{i=1}^n X_i$, then

\begin{align*}
 \Pr\Big[\bar{X} - E[\bar{X}] \geq \kappa\Big] &\leq 
  \exp\left( - \frac{2\kappa^2n^2}{\sum_{i=1}^n (b_i-a_i)^2}\right)
\end{align*}
where $\exp(x) = e^x$. When $0 \leq X_i \leq 1$, this gives 
(by symmetry for \eqref{SymmetricHoeffding} and summation for \eqref{ValAbsHoeffding})
\begin{align}
 \Pr\Big[\bar{X} - E[\bar{X}] \geq  \phantom{-}\kappa\Big] &\leq 
   \exp\left( - 2\kappa^2 n\right) \nonumber \\
 \Pr\Big[\bar{X} - E[\bar{X}] \leq  -\kappa\Big] &\leq 
   \exp\left( - 2\kappa^2 n\right) \label{SymmetricHoeffding} \\
 \Pr\Big[\big|\bar{X} - E[\bar{X}]\big| \geq \kappa\Big] &\leq 2\exp\left( - 2\kappa^2 n\right). \label{ValAbsHoeffding}
\end{align}

\end{theorem}

% ==================================================================
% ==================================================================

\urlstyle{rm}
%\bibliography{bgkm}

%%% RG - Automatically created by bibtex. Do not change manually.

\end{document}